%% file: root.tex
\documentclass[journal,twoside, web]{ieeecolor}
\usepackage{generic}
\usepackage[pdftex]{graphicx}
\usepackage[outdir=./]{epstopdf}

\title{Recovery of Localization Errors in Sensor Networks \mbox{using Inter-Agent Measurements}}

\author{Shiraz Khan and Inseok Hwang 
\thanks{\textcolor{blue}{
© 2023 IEEE. Personal use of this material is permitted. Permission from IEEE must be obtained for all other uses, in any current or future media, including reprinting/republishing this material for advertising or promotional purposes, creating new collective works, for resale or redistribution to servers or lists, or reuse of any copyrighted component of this work in other works.
}}
\thanks{This research is funded by the Secure Systems Research Center (SSRC) at the Technology Innovation Institute (TII), UAE. The authors are grateful to Dr. Shreekant (Ticky) Thakkar and his team members at the SSRC for their valuable comments and support.}%
\thanks{The authors are with School of Aeronautics and Astronautics, Purdue University,
West Lafayette, IN 47906
        (Email: {\tt\small shiraz@purdue.edu, ihwang@purdue.edu})}%
}
\usepackage{amsmath,amssymb,amsfonts, mathtools, xfrac}
\usepackage{algorithm,algorithmic}

\usepackage{cite}  

\newtheorem{theorem}{Theorem}
\newtheorem{lemma}{Lemma}
\newtheorem{proposition}{Proposition}

\newtheorem{remark}{Remark}

\DeclareMathOperator{\rank}{rank}
\DeclareMathOperator{\vspan}{span}

\usepackage{hyperref}
\hypersetup{hidelinks=true}
\usepackage{textcomp}
\def\BibTeX{{\rm B\kern-.05em{\sc i\kern-.025em b}\kern-.08em
    T\kern-.1667em\lower.7ex\hbox{E}\kern-.125emX}}
\begin{document}
\maketitle
\begin{abstract}
A practical challenge which arises in the operation of sensor networks is the presence of sensor faults, biases, or adversarial attacks, which can lead to significant errors incurring in the localization of the agents, thereby undermining the security and performance of the network.
We consider the problem of identifying and correcting the localization errors
using inter-agent measurements, such as the distances or bearings from one agent to another, which can serve as a redundant source of information about the sensor network's configuration.
The problem is solved by searching for a block sparse solution to an underdetermined system of equations, where the sparsity is introduced via the fact that the number of localization errors is typically much lesser than the total number of agents. Unlike the existing works, our proposed method does not require the knowledge of the identities of the \textit{anchors}, i.e., the agents that do not have localization errors.
We characterize the necessary and sufficient conditions on the sensor network configuration under which a given number of localization errors can be uniquely identified and corrected using the proposed method. 
The applicability of our results is demonstrated numerically by processing inter-agent distance measurements using a sequential convex programming (SCP) algorithm to identify the localization errors in a sensor network.
\end{abstract}

\input{sections/1_introduction}
\input{sections/1_notation}
\input{sections/2_formulation}
\input{sections/3_optimization}
\input{sections/4_recoverability}
\input{sections/5_robustness}
\input{sections/6_simulation}

\section*{References}
\bibliographystyle{ieeetr}
\bibliography{IEEEabrv,refs} 
\input{sections/7_appendices}
\end{document}

%% file: sections/1_introduction.tex
\section{INTRODUCTION}
Groups of autonomous systems such as rovers and uncrewed aerial vehicles (UAVs) can accomplish complex tasks in a collaborative manner.
In order to facilitate this collaboration, they are typically interconnected in terms of their sensing and/or communication capabilities, forming a sensor network. For instance, an agent of a sensor network may be able to measure its \textit{distances} from the other agents in its vicinity using received signal strength (RSS) or time of arrival (ToA) based ranging \cite{cao2021relative, sadowski2018rssi}. Similarly, cameras pointed from one vehicle to another as well as the angle of arrival (AoA) of inter-agent communications are examples of \textit{bearing} measurements which may be available between agents \cite{deghat2014localization, zhang2020bearing_aoa}. The design of efficient algorithms which utilize distance and bearing measurements to enhance the localization (i.e., position estimation) capabilities of sensor networks has been well-studied in the literature, using the theories of distance and bearing rigidity \cite{decent_dist_rigidity_2015,zhao2016localizability,zhao2019bearing}. Similarly, the theory of weak rigidity was developed recently in \cite{jing2018weak}, which can model sensor networks in which the agents are able to measure the subtended angles of their neighbors, e.g., using the time difference of arrival (TDoA) of communications or scanning radars \cite{subtended_angle_scan_TDOA}. 

The academic and practical importance of inter-agent measurements stems from the observation that the onboard localization 
mechanisms of autonomous systems can incur significant errors. This is because the sensors which contribute to the position estimation, including satellite-based geopositioning systems, are susceptible to faults (such as bias, loss of signal, etc.) and cyberattacks (such as sensor spoofing and replay attacks) \cite{multipath_2013, ranyal2021unmanned}. Motivated by these concerns, several researchers have proposed the use of inter-agent measurements for detecting and identifying the localization errors \cite{xia2013distributed, lim2019detecting}. Inter-agent measurements are especially useful for identifying localization errors in low-cost and resource-constrained sensor network applications, in which it may not be possible to install expensive and/or heavy onboard components such as LiDAR sensors. 
Despite the practical importance and relevance of this problem, a formal theoretical treatment of it is lacking in the literature. Most theoretical works have focused on directly using the inter-agent measurements for onboard localization \cite{Eren_rigidity_randomness_2004, decent_dist_rigidity_2015,zhao2016localizability,zhao2019bearing, zhang2020bearing_aoa}, rather than addressing the question of how these measurements can be used to identify faults or cyberattacks. 
In particular, they assume that a subset of the agents, called as \textit{anchors} or \textit{beacons} in the literature, are known to be correctly localized, precluding the possibility that the anchors themselves may have localization errors; in this study, we drop the assumption that the set of correctly localized agents is known \textit{a priori}.
 
Generally, the problem of fault and/or cyberattack identification in multi-sensor systems is solved using combinatorial approaches such as fault tree analysis, or by using a bank of observers
\cite{faultdetection_2012}\cite{yang2018multi}, which are intractable for large-scale sensor networks as their computational complexity
scales exponentially with respect to the size of the network. The authors in \cite{nstacked_2015} showed that the computational complexity of fault or cyberattack detection can be reduced by
reformulating it as an $l_0$ minimization\footnote{In the literature, this is sometimes called $l_0$ ``norm" minimization; we avoid this usage as the $l_0$ operation does not satisfy the axioms of a norm
\cite{nstacked_2015}.} problem, where the $l_0$ operation, as defined, counts the number of non-zero elements of a vector.
In addition to having a lower computational complexity, the $l_0$ minimization approach is able to uniquely recover the solution of an underdetermined system of equations.
This has motivated researchers to make the additional assumption on the sparsity of the error vector (i.e., they assume that the number of faults/cyberattacks is small) enabling the use of $l_0$  minimization (or when the solution is known to be a block vector, $l_2/l_0$  minimization \cite{latushkin_null_2015}) for the detection and identification of sparse sensor faults and cyberattacks \cite{ozay_sparse_2013,nstacked_2015,attack_sparse_Lu_Yang_2022}. 

Nevertheless, the $l_0$ minimization problem is still computationally expensive due to its combinatorial nature, which motivates the search for an efficient algorithm that scales well (e.g., linearly) with respect to the number of agents in the sensor network. Inspired by the abundant literature on compressive sensing, the authors of \cite{nstacked_2015} relaxed the $l_0$ minimization problem to an $l_1$ minimization problem, which is a convex optimization problem that admits efficient algorithms for finding its solution. However, $l_1$ (or in the case of block vectors, $l_2/l_1$) minimization can only recover a given number of errors uniquely and exactly when one of certain conditions is met, such as the (block) restricted isometry property or the (block) null space property (NSP) \cite{latushkin_null_2015, robust_NSP_2017, Block_COSAMP_2019}. These conditions are often ignored or assumed to hold true in the literature, as they depend on the measurement matrix and are difficult to establish \textit{a priori}, especially when the measurement matrix is a part of the problem specification (i.e., cannot be chosen arbitrarily) \cite{ozay_sparse_2013, yuan2020gps}.

In this paper, we propose a novel method for processing the inter-agent distance or bearing measurements to recover (i.e., identify and reconstruct) the localization errors in $2$ and $3$-dimensional sensor network configurations, while addressing the aforementioned shortcomings of the existing works.
The proposed method is developed by reformulating the error recovery problem as an $l_2/l_q$ minimization problem, where $0\leq q \leq 1$. While the existing solutions for localization error identification assume that some of the agents are classified as anchors \textit{a priori}, the proposed $l_2/l_q$ minimization approach does not require this assumption; to our knowledge, the possibility of dropping the foregoing assumption has not been recognized or explored in the literature previously.
Furthermore, we establish the conditions on the sensor network's connectivity and configuration under which a given number of localization errors can be uniquely identified and corrected using the proposed $l_2/l_q$ minimization method. 
%
Our analysis is based on the conditions for the recovery of block sparse signals borrowed from the literature on compressive sensing, which manifest as purely geometric properties of the sensor network's configuration. On the other hand, results from distance and bearing rigidity theory are used to establish the sufficient conditions on the sensor network's connectivity.
To demonstrate the applicability of our results, a sequential convex programming (SCP) algorithm is proposed for solving the $l_2/l_1$ minimization problem efficiently, while accommodating the nonlinearity in the measurement model.
The proposed SCP algorithm is validated using a numerical example of identifying the localization errors in a sensor network using inter-agent distance measurements.

%% file: sections/1_notation.tex
The paper is organized as follows; Section \ref{sec:problem_formulation} presents the mathematical descriptions of the sensor network, localization errors, and inter-agent measurements. In Section \ref{sec:optimization}, the main ideas from compressive sensing are introduced and motivated. Section \ref{sec:recoverability} establishes the conditions for the recoverability of localization errors using inter-agent measurements in the noise-free case. Section \ref{sec:robustness} extends our results to accommodate practical considerations such as measurement noise, imperfect estimates, and linearization error, and presents the SCP algorithm to solve the $l_2/l_1$ minimization problem in a robust manner.
Finally, the applicability of our approach is demonstrated through a numerical example in Section \ref{sec:numerical}.

\textit{Notation:} 
Given a block vector $\bold v\in\mathbb R^{dn}$ which is partitioned into $n$ blocks of length $d$ each, $\bold v[i]$ refers to the $i^{th}$ block of $\bold v$. The $l_q$ norm of $\bold v$ is denoted by $\|\bold v\|_{q}$.
Similar to \cite{efficient_block_sparse_2010, wang_wang_xu_2013, latushkin_null_2015, robust_NSP_2017}, we define the following notation:
\[
\|\bold v\|_{2,q}=
\begin{cases}
\begin{array}{lc}
     \sum_{i=1}^{n} \mathbb I\Big(\|\bold v[\small i]\|_2>0\Big) \quad & q=0 \vspace{2pt}\\
     \Big( \sum_{i=1}^{n} \|\bold v[\small i]\|_2^q\Big)^{1/q} & 0<q<\infty \vspace{2pt} \\
     \max_{1\leq i\leq n} \big(\|\bold v[\small i]\|_2\big) & q = \infty,
\end{array}
\end{cases}
\]
where $\mathbb I(\|\bold v[\small i]\|_2>0)$ is the indicator function which is equal to $1$ when $\|\bold v[\small i]\|_2>0$ and $0$ otherwise.
As defined, $\|\bold v\|_{2,0}$ counts the number of non-zero blocks of $\bold v$, referred to as its \textit{block sparsity}. 
Given an index set $\mathcal S\subset \{1, 2, \dots, n\}$, $\mathcal S^\complement$ denotes its complement, $\{1, 2, \dots, n\}\backslash \mathcal S$. $\bold v_{\mathcal S}$ is the vector whose support is restricted to the blocks corresponding to $\mathcal S$, i.e., 
\[\bold v_{\mathcal S}[i] = \begin{cases}
\begin{array}{ll}
\bold v[i] \quad & i\in \mathcal S\\
\bold 0 & i\in \mathcal S^\complement,
\end{array}
\end{cases}
\]
where $\bold 0$ denotes a vector of zeros. $\bold 1_d$ denotes a vector of ones of length $d$ and $\bold I_d$ is the $d\times d$ identity matrix. For a matrix $\bold A$, we let 
$\ker(\bold A)$ denote its kernel or null space. 
Lastly, `$\otimes$' is the Kronecker product operation for vectors and matrices.

%% file: sections/2_formulation.tex
\section{PROBLEM FORMULATION}
\label{sec:problem_formulation}
\subsection{Sensor Network Model}
Let $\mathcal G = (\mathcal V, \mathcal E)$ be an undirected graph, where the vertices $\mathcal V=\lbrace 1, 2, \dots, |\mathcal V|\rbrace$ represent the sensor network agents and the edges $\mathcal E \subseteq \mathcal V \times \mathcal V$ represent the availability of inter-agent (distance or bearing) measurements. 
The position of each agent is represented by a vector in $\mathbb R^{d}$, where the dimension $d$ is either $2$ or $3$ depending on the application.
The collective configuration of the sensor network can be represented by the block vector
\begin{equation}
\bold p = \big[\ \bold p[1]^\top\ \bold p[2]^\top\ \bold p[3]^\top\ \dots\ \bold p[\text{\footnotesize{$|\mathcal V|$}}]^\top\ \big]^\top \in \mathbb R^{d|\mathcal V|},
\end{equation}
where $\bold p[i]\in\mathbb R^d$ is the $i^{th}$ agent's position. We denote the components of $\bold p[i]$ as $\bold p[i] = \big[\bold p[i]_1\ \bold p[i]_2\big]^\top$ when $d=2$, and similarly for $d=3$. We assume that two agents' positions cannot coincide. 
While the results of this paper can be extended to the case where $\bold p$ evolves according to a dynamical model, we discuss the static case (where $\bold p$ is fixed) in order to keep the analysis and presentation concise.

\subsection{Position Estimates}
Each agent uses a set of onboard sensors for localization (i.e., the estimation of its position). The estimated positions of the agents are collectively represented by the block vector $\hat{\bold p} \in \mathbb R^{d|\mathcal V|}$, such that $\hat{\bold p}[i] \in \mathbb R^d$ is the $i^{th}$ agent's position estimate. Let $\mathcal D\subseteq \mathcal V$ be the set of agents that have localization errors, e.g., due to sensor faults or spoofing attacks.
To keep the analysis succinct, for now we assume that the agents in $\mathcal D^\complement$ have perfect estimates, such that $\hat {\bold p}_{\mathcal D^\complement} = \bold p _{\mathcal D ^\complement}$, and
\[\|\hat{\bold p}[i] - \bold p[i]\|_2 > 0 \ \Leftrightarrow\  \ i\in \mathcal D\]
%
We define $ \bold x \coloneqq \bold p - \hat{\bold p}$ as the vector of localization errors. Thus, $\bold x$ has exactly $|\mathcal D|$ non-zero blocks, i.e.,
$\|\bold x\|_{2,0} = |\mathcal D|$, and
$\bold x$ is said to be \textit{block $|\mathcal D|$-sparse}. In Section \ref{sec:robustness}, we extend the above formulation to incorporate the case where the agents in $\mathcal D^\complement$ have imperfect position estimates.

\begin{remark}
One can also consider an alternative problem in which the agents in $\mathcal D$ are maliciously misreporting their positions to the other agents (via broadcast-based communication), so that $\hat {\bold p}[i] \neq \bold {p}[i]\ \forall i \in \mathcal D$. In either problem, the block vector $\bold x$ needs to be reconstructed using the inter-agent measurements, where the non-zero blocks of $\bold x$ correspond to the agents in $\mathcal D$.
Hence, either problem can be solved using the proposed approach described in the remainder of this paper, due to the symmetric relationship between the two problems.
\end{remark}

\subsection{Distance and Bearing Measurements}
\label{sec:setup_measurements}
The agents adjacent to each other in $\mathcal G$ are able to obtain inter-agent measurements which can be used to reconstruct the localization error vector, $\bold x$. Given a function $\boldsymbol \phi:\mathbb R^d \times\mathbb R^d\rightarrow \mathbb R^m$, the inter-agent measurement model is defined as
\begin{equation}
\bold y_{ij} = \boldsymbol \phi(\bold p[i], \bold p[j]) + \bold e_{ij} \quad \forall (i,j)\in \mathcal E
\label{eq:edge_measurement}
\end{equation}
where $\bold e_{ij}\in \mathbb R^m$ is a bounded noise term.
Equation (\ref{eq:edge_measurement}) can be expressed in the block vector form as
\begin{equation}
    \bold y = \boldsymbol \Phi(\bold p) + \bold e
    \label{eq:edge_measurement_block}
\end{equation}
where $\bold y=[\bold y_{ij}]$, $\bold e=[\bold e_{ij}]$, and $\boldsymbol \Phi(\bold p) = \big[\boldsymbol \phi(\bold p[i], \bold p[j])\big]$ are block vectors.

We consider two types of measurements having the form of (\ref{eq:edge_measurement}) that commonly arise in sensor network applications, namely distance and bearing measurements. In the case where the agents can measure their distances from each other, we have 
\begin{equation}
\boldsymbol \phi_D(\bold p[i], \bold p[j])\coloneqq\frac{1}{2}\|\bold p[i]-\bold p[j]\|_2^2
\label{eq:distance_model}
\end{equation}
where the constant $\sfrac{1}{2}$ is introduced to keep the forthcoming notation concise. 
Let $\bold \Phi_D(\bold p) = \big[\boldsymbol \phi_D(\bold p[i], \bold p[j])\big]
$ denote the block vector of length $|\mathcal E|$ comprising all the inter-agent distance functions. The \textit{distance rigidity matrix} \cite{trinh2016further}, denoted by $\bold R_D(\bold p)$, is then defined as 
\begin{equation}
\bold R_D(\bold p)\coloneqq \nabla \bold \Phi _D({\bold p})\in R^{|\mathcal E|\times d|\mathcal V|},
\end{equation}
which is the Jacobian of the measurement vector.
Thus, the $k^{th}$ row of $\bold R_D(\bold p)$ corresponds to the $k^{th}$ edge of the graph and is of the form
\[
\begin{bmatrix} \ 0 & \dots & 0 & (\bold p[i] - \bold p[j])^\top & \dots & (\bold p[j] - \bold p[i])^\top & \dots\  \end{bmatrix}
\]
\noindent where $(i,j)$ is the $k^{th}$ edge in $\mathcal E$. Similarly, in the case of bearing measurements, we define 
\begin{equation}
\boldsymbol \phi_B(\bold p[i], \bold p[j])\coloneqq\frac{\bold p[i] - \bold p[j]}{\|\bold p[i] - \bold p[j]\|_2}
\end{equation}
which are the inter-agent unit vectors.
In this case, the Jacobian $\nabla \bold \Phi_B(\bold p)$ is denoted by $\bold R_B(\bold p) \in \mathbb R^{d|\mathcal E|\times d|\mathcal V|}$ and is called the \textit{bearing rigidity matrix} \cite{zhao2019bearing} whose $k^{th}$ row is of the form:
\begin{align*}
\begin{bmatrix} & 0 & \dots & 0 & \frac {\bold P_{ij}(\bold p)}{\|\bold p[i] - \bold p[j]\|_2} & \dots & \frac{\bold P_{ji}(\bold p)}{\|\bold p[i] - \bold p[j]\|_2} & \dots & \end{bmatrix},
\end{align*}
where $(i,j)$ is the $k^{th}$ edge in $\mathcal E$, and 
\begin{equation}
\bold P_{ij}(\bold p)=\bold I_d - 
\text{\footnotesize{
$\frac{(\bold p[i] - \bold p[j])(\bold p[i] - \bold p[j])^\top}{\|\bold p[i] - \bold p[j]\|_2^2}$
}}
\end{equation}
is the projection matrix corresponding to the subspace orthogonal to $\bold p[i] - \bold p[j]$.

In real-time applications, such as when using an Extended Kalman Filter (EKF) estimator, the Jacobian matrix is evaluated at the estimated state instead of the true state. For example,
as the squared Euclidean norm $\|\cdot \|_2^2$ is continuously differentiable everywhere, we can use the Taylor series approximation of $\bold \Phi_D(\bold p)$ about $\hat {\bold p}$ to establish
\begin{align} 
\bold \Phi_D(\bold{p}) \approx \bold \Phi_D(\hat{\bold p}) + \bold R_D(\hat {\bold p})(\bold p - \hat{\bold p})
\label{eq:taylor_series}
\end{align}
where the approximation error, which is on the order of $\| \hat{\bold p} - \bold p \|^2_2$, is assumed to be small. This is equivalent to the assumption that $\bold x$ is bounded, i.e., the agents have not deviated too far from their estimated positions. In Section \ref{sec:robustness}, we discuss how the linearization error in (\ref{eq:taylor_series}) can be compensated by using a bootstrapping approach.
%

%% file: sections/3_optimization.tex
\section{Block-Sparse Optimization Problem}%
\label{sec:optimization}
\noindent
Consider the residual vector $\bold z$, defined as 
\begin{equation}
    \bold z \coloneqq \bold y - \bold \Phi(\bold {\hat p}) 
\end{equation}
Using (\ref{eq:edge_measurement_block}) and (\ref{eq:taylor_series}), we have
\begin{align}
\bold z = \bold R (\bold p - \bold {\hat p}) +\bold e &= \bold R \bold x + \bold e
\end{align}
where the notation $\bold R$ is introduced for brevity;
it denotes either the distance or bearing rigidity matrix ($\bold R_D(\bold p)$ or $\bold R_B(\bold p)$, respectively) depending on the context. 

A naive approach for recovering the localization errors in the noise-free case (i.e., the case where $\bold e=\bold 0$) is to solve the system of linear equations
$\bold z = \bold R\hat{\bold x}$ for $\hat{\bold x}$. However, there are two issues with this approach:
\begin{itemize}
    \item \textit{Non-Uniqueness of the Solution}: Both the distance and bearing rigidity matrices \big($\bold R_D$ and $\bold R_B$, respectively\big) have non-trivial null spaces (as shown in Section \ref{sec:recoverability}), irrespective of the configuration and connectivity of the sensor network. Thus, the system of linear equations $\bold z = \bold R\hat{\bold x}$ does not have a unique solution; it is said to be underdetermined.
    \item \textit{Non-Sparsity of the Solution}: For $0<s<|\mathcal V|$, the set of all block $s$-sparse vectors in $\mathbb R^{d|\mathcal V|}$ is a union of lower-dimensional subspaces, each having the dimension $ds$ \cite{Eldar_Mishali_2009}. Consequently, the set of block $s$-sparse vectors has a measure of $0$ in the ambient space $\mathbb R^{d|\mathcal V|}$, meaning that general iterative algorithms used to solve the foregoing system of equations will yield non-sparse solutions with overwhelming probability. This is equivalent to the classification of all the agents in $\mathcal V$ as having localization errors (i.e., $\mathcal D = \mathcal V$), which may be undesirable.
\end{itemize}
In sensor network applications where it is more likely that $|\mathcal D|<|\mathcal V|$, under certain conditions on the block-sparsity of $\bold x$, it is possible to uniquely recover it by solving the following optimization problem: 
\begin{align*}
\begin{array}{rclc}
   \textsc{P1:}  &  &
\begin{array}{rl}
\underset{\hat{\bold x}}{\textnormal{minimize}}\quad  
&\| \hat{\bold x}\|_{2,0}  \\
\textnormal{subject to} \quad
& \bold z = \bold R\hat{\bold x}
\end{array} & \qquad \qquad
\end{array}
\end{align*}
which we call the (constrained) $l_2/l_0$ minimization problem. The function $\|\cdot\|_{2,0}$ is not convex, however, and $l_2/l_0$ minimization is known to be a computationally hard problem \cite{robust_NSP_2017}.
Therefore, it is desirable to study the recoverability of $\bold x$ using the following relaxed $l_2/l_q$ minimization problem, where $0<q\leq1$:
\begin{align*}
\begin{array}{rclc}
   \textsc{P2:}  &  &
\begin{array}{rl}
\underset{\hat{\bold x}}{\textnormal{minimize}}\quad  
&\| \hat{\bold x}\|_{2,q}  \\
\textnormal{subject to} \quad
& \bold z = \bold R\hat{\bold x}
\end{array} & \qquad \qquad
\end{array}
\end{align*}
In particular, $l_2/l_1$ minimization is a convex optimization problem and can be solved using second-order cone programming (SOCP) solvers \cite{wang_wang_xu_2013}, block variants of basis pursuit and matching pursuit solvers \cite{efficient_block_sparse_2010, giaralis_bomp_2021, Block_COSAMP_2019}, as well as distributed algorithms \cite{mota2011basis}. The nonlinearity and noise in the measurement model can also be accommodated by using basis pursuit denoising \cite{Eldar_Mishali_2009}. This makes the $l_2/l_q$ minimization problem (\textsc{P2}) an attractive alternative to the $l_2/l_0$ version (\textsc{P1}). 
Moreover, it has been noted that $l_2/l_q$ optimization outperforms $l_q$ optimization when the solution is known to be block sparse \cite{Eldar_Mishali_2009,efficient_block_sparse_2010}. 

%% file: sections/4_recoverability.tex
\section{Conditions for Recoverability}
\label{sec:recoverability}
In this section, we derive conditions on the sensor network $\mathcal G$, its configuration $\bold p$, and the error vector $\bold x$, under which one can uniquely and exactly recover $\bold x$ by solving either the $l_2/l_0$ minimization problem P1, or the $l_2/l_q$ minimization problem P2 (where $0<q\leq 1$). The noise-free case is discussed first, which gives us qualitative conditions for recoverability, whereas Section \ref{sec:robustness} incorporates quantitative considerations such as measurement noise.

\subsection{$l_2/l_0$ Minimization}
From the literature on compressive sensing, we have the following lemma about the recoverability of $\bold x$ using the $l_2/l_0$ minimization problem P1.

\begin{lemma}[$l_2/l_0$ Recoverability \cite{afdideh2016recovery}]
Let $\bold x$ be a solution to the $l_2/l_0$ minimization problem P1. Then, $\bold x$ is the unique solution to P1 if
\begin{equation}
    \|\bold x\|_{2,0} < \hspace{1pt}\frac{1}{2}\hspace{3pt}{
            \min \left\lbrace \|\bold v\|_{2,0} \hspace{1pt}\big\vert\hspace{1pt}
                 \bold v \in \ker(\bold R), \bold v\neq \bold0
                 \right\rbrace
        }
\label{eq:block-spark}
\end{equation}
\label{lemma:l0_recoverability}
\end{lemma}

Note that Lemma \ref{lemma:l0_recoverability} establishes the uniqueness of the solution of P1 in spite of $\bold R$ having a non-trivial null space.
The quantity $\min \left\lbrace \|\bold v\|_{2,0} \hspace{1pt}\big\vert\hspace{1pt}
                 \bold v \in \ker(\bold R), \bold v\neq \bold0
                 \right\rbrace$ 
                 is called the \textit{block spark} of the matrix $\bold R$ \cite{afdideh2016recovery}.
Although it is difficult to compute in general, we can characterize the block spark of a well-studied class of rigidity matrices, namely those that have maximal rank, using the following lemmas. In what follows, given two subspaces $V_1, V_2\subseteq \mathbb R^n$, $V_1+V_2$ refers to the sum $\{\bold v_1+\bold v_2\big|\bold v_1\in V_1, \bold v_2\in V_2\}$ which is also a subspace. Given integers $n$ and $k$, $\binom{n}{k}$ denotes the corresponding binomial coefficient.

\begin{lemma}[Null Space of $\bold R_D$]
For $d=2$ or $3$, let $S_d$ be the space of all $d\times d$ real skew-symmetric matrices. We have 
\begin{equation}
\left\lbrace \bold 1_{|\mathcal V|}\otimes \bold q 
\hspace{1pt}\big|\hspace{1pt}
\bold q \in \mathbb R^d \right\rbrace
+ \left\lbrace(\bold I_{|\mathcal V|} \otimes \bold S)\bold p
\hspace{1pt}\big|\hspace{1pt}
\bold S \in S_d\right\rbrace 
\subseteq \ker\left(\bold R_D\right)
\label{eq:rD_null_space}
\end{equation}
As a consequence, $\rank(\bold R_D)\leq d|\mathcal V| - \binom{d+1}{2}$. \vspace{2pt}
\label{lemma:distance_rigidity_properties}
\end{lemma}
\begin{proof}
Consider the case of $d=3$.
The first summand on the left hand side of (\ref{eq:rD_null_space}) contains block vectors of the form
\begin{equation}
\bold v = \begin{bmatrix}
\bold q^\top & \bold q^\top & \dots & \bold q^\top    
\end{bmatrix}^\top
\end{equation}
for some vector $\bold q \in \mathbb R^3$.
It can be checked that, for block vectors of this form,  $\bold R_D \bold v = \bold 0$.
Focusing on the second summand of (\ref{eq:rD_null_space}), we can observe that $S_3$ is a vector space, as it is closed under the scalar multiplication and matrix addition operations.
Consider the following basis for $S_3$:
\begin{equation}
    \left\lbrace
    \begin{bmatrix*}[r] 0 & 1 & 0\\ -1 & 0 & 0 \\ 0 & 0 & 0 \end{bmatrix*},\begin{bmatrix*}[r] 0 & 0 & 1\\ 0 & 0 & 0 \\ -1 & 0 & 0 \end{bmatrix*},\begin{bmatrix*}[r] 0 & 0 & 0\\ 0 & 0 & -1 \\ 0 & 1 & 0 \end{bmatrix*}
    \right\rbrace
\end{equation}
Multiplying each basis element with the corresponding block of $\bold p$, we see that $\left\lbrace
(\bold I_{|\mathcal V|} \otimes \bold S)\bold p
\hspace{1pt}\big|\hspace{1pt}
\bold S \in  S_3\right\rbrace$ 
is composed of block vectors $\bold v$ that have the following property: 
for some $a,b,c \in \mathbb R$, 
\begin{equation}
    \bold v [i] = a
    \begin{bmatrix*}[r]
        \bold p[i]_y \\ -\bold p[i]_x \\  0\quad
    \end{bmatrix*} + b \begin{bmatrix*}[r]
        \bold p[i]_z \\  0\quad \\ -\bold p[i]_x 
    \end{bmatrix*} + c\begin{bmatrix*}[r]
         0\quad \\ -\bold p[i]_z \\ \bold p[i]_y
    \end{bmatrix*}\ \forall i\in\mathcal V
\end{equation}
Once again, it can be verified that $\bold R_D \bold v = \bold 0$.
A similar reasoning applies to the $d=2$ case as well.
Finally, the condition on the rank follows from the dimension of the null space in either case, which is $\binom{d+1}{2}$.
\end{proof}
\begin{remark}
The intuition for Lemma \ref{lemma:distance_rigidity_properties} is that the isometries (i.e., rigid motions) of $\mathbb R^d$ preserve the distances between the sensor network agents. Of these isometries, there are $d$ translations and $\binom{d}{2}$ rotations, corresponding to the first and second summand of the left hand side of (\ref{eq:rD_null_space}), respectively, which add up to the $\binom{d+1}{2}$ dimensions of $\ker(\bold R_D)$. Furthermore, the skew-symmetric matrices $S_d$ make up the Lie algebra $\frak{s0}(d)$ corresponding to the infinitesimal rotations in $\mathbb R^d$, which explains their relevance to distance rigidity.
\label{rem:R_D_nullspace}
\end{remark}

The null space of $\bold R_D$ was characterized in \cite{trinh2016further} and \cite{decent_dist_rigidity_2015} for the cases of $d=2$ and $3$, respectively, whereas the characterization given in Lemma \ref{lemma:distance_rigidity_properties} is a generalization which  holds true in arbitrary dimensions, i.e., $d\geq 1$. An equivalent result about the bearing rigidity matrix, $\bold R_B$, is as follows.
\begin{lemma}[Null Space of $\bold R_B$\protect{\cite[Lemma 3]{Zhao_Zelazo_2016}}]
    For $d=2$ or $3$, we have
    \begin{equation}
    \left\lbrace \bold 1_{|\mathcal V|}\otimes \bold q 
\hspace{1pt}\big|\hspace{1pt}
\bold q \in \mathbb R^d \right\rbrace+\vspan\big(\{\bold p\}\big)\subseteq \ker\left(\bold R_B\right),
    \label{eq:rb_nullspace}
    \end{equation}
    and consequently, $\rank(\bold R_B)\leq d|\mathcal V| - d - 1$.
    \label{lemma:bearing_rigidity_properties}
\end{lemma}
\begin{proof}
    The proof is similar to that of Lemma \ref{lemma:distance_rigidity_properties}; it
    involves checking that the vectors in each summand of (\ref{eq:rb_nullspace}) vanish under multiplication with $\bold R_B$.
\end{proof}
\begin{remark}
    The second summand in (\ref{eq:rb_nullspace}) corresponds to scaling (expansion and shrinking) of $\mathbb R^d$, which preserves the orientations of the vectors $\bold p[i] - \bold p[j]$, $\forall i,j\in \mathcal V$, and therefore preserves the inter-agent bearing measurements.
    \label{rem:R_B_nullspace}
\end{remark}

Thus, given a sensor network with $|\mathcal V|$ agents, $\bold R$ has a maximal rank (across all possible values of $\mathcal E$ and $\bold p$) that is strictly less than the number of columns in $\bold R$. A sensor network configuration is said to be \textit{infinitesimally rigid} in dimension $d$ when the corresponding rigidity matrix has maximal rank, in which case, the symbol `$\subseteq$' in (\ref{eq:rD_null_space}) and (\ref{eq:rb_nullspace}) become equalities \cite{decent_dist_rigidity_2015,zhao2019bearing}.
There is extensive literature discussing how infinitesimal rigidity can be maintained in a distributed manner \cite{Olfati-Saber_Murray_2002,decent_dist_rigidity_2015,bearing_rigidity_maintenance_2017}. It can also be guaranteed with high probability in random sensor network configurations (modeled as random geometric graphs) based on the spatial density and sensing radius of the agents \cite{Eren_rigidity_randomness_2004}. 

We now state the condition for $l_2/l_0$ recoverability (i.e., the ability to recover the localization error vector $\bold x$ uniquely by solving the $l_2/l_0$ minimization problem P1) for infinitesimally rigid sensor network configurations.
\begin{theorem}
Suppose $\bold R$ has maximal rank. Then, a solution $\bold x$ of the $l_2/l_0$ minimization problem P1 is the unique solution to P1 if 
\begin{equation}
\|\bold x\|_{2,0}<\frac{1}{2}\left(|\mathcal V|-\tilde s\right), 
\end{equation}
where $\tilde s$ is determined as follows:
\begin{enumerate}
    \item if $\bold R=\bold R_D$ and $d=2$, then $\tilde s=1$
    \item if $\bold R=\bold R_D$ and $d=3$, then $\tilde s$ is equal to the maximum number of colinear points in $\{\bold p[i]\}_{i\in \mathcal V}$
    \item if $\bold R=\bold R_B$ and $d=2$ or $3$, then $\tilde s=1$
\end{enumerate}
\label{theorem:l0_recoverability}
\end{theorem}
\begin{proof}
We provide a sketch of the proof.
First, we need to show that
\begin{equation}
     \hspace{3pt}{
            \min \left\lbrace \|\bold v\|_{2,0} \hspace{1pt}\big\vert\hspace{1pt}
                 \bold v \in \ker(\bold R), \bold v\neq \bold0
                 \right\rbrace
        } = |\mathcal V| - \tilde s \hspace{3pt}
\end{equation}
which can be verified on a case-by-case basis using Lemmas \ref{lemma:distance_rigidity_properties} and \ref{lemma:bearing_rigidity_properties}. Thereafter, Lemma \ref{lemma:l0_recoverability} establishes that $\bold x$ is the unique solution to P1 in each case.

While the proof can be completed algebraically,
the following is an intuitive explanation:
$\tilde s$ is equal to the number of fixed points of the corresponding transformations of $\mathbb R^d$ (see Remarks \ref{rem:R_D_nullspace} and \ref{rem:R_B_nullspace}), which is $1$ for $2$D rotations, the number of points lying on the axis of rotation for $3$D rotations, and $1$ for scaling; these correspond to cases $1$, $2$ and $3$ of the theorem, respectively.
\end{proof}


Recall that $|\mathcal D|$ is equal to $\|\bold x\|_{2,0}$, the block sparsity of $\bold x$.  Hence, in each of the cases considered in Theorem \ref{theorem:l0_recoverability},
roughly up to half of the agents in $\mathcal V$ can have localization errors before the $l_2/l_0$ minimization approach fails to recover the error vector $\bold x$ uniquely.
Note that the existing literature on rigidity theory assumes that the identities of the anchors (and consequently, the support of $\bold x$) is known \cite{zhao2019bearing}, whereas Theorem \ref{theorem:l0_recoverability} pertains to the more general case where the support of $\bold x$ may not be known.

\subsection{$l_2/l_q$ Recoverability}
We now turn our attention to the $l_2/l_q$ minimization P2, for $0<q\leq 1$, where $q=1$ is the most attractive approach due to its computational ease.
A necessary and sufficient condition for $l_2/l_q$ minimization to uniquely recover a block sparse solution is given via the $l_q$ block null space property (NSP) \cite{latushkin_null_2015}, which is defined as follows.
A matrix $\bold A\in\mathbb R^{m\times dn}$ is said to satisfy the $l_q$ block NSP of order $s$ if, for all subsets $\mathcal S\subseteq \{1, 2, \dots, n\}$ with $|\mathcal S|\leq s$, we have
\begin{equation}
\|\bold v_{\mathcal S}\|_{2,q} < \|\bold v_{\mathcal S^\complement}\|_{2,q},\quad \forall\hspace{1pt}\bold v\in \ker(\bold A)\backslash \mathcal \{\bold 0\}
\label{eq:block_nsp_condition}
\end{equation}
\begin{lemma}
 [$l_2/l_q$ Recoverability]
Let $\bold x$ be a solution to the $l_2/l_q$ minimization problem P2, where $0<q\leq 1$. Then, $\bold x$ is the unique solution to P2
if and only if $\bold R$ satisfies the $l_q$ block NSP of order $\|\bold x\|_{2,0}$.
\label{lemma:lq_recoverability}
\end{lemma}
\begin{proof}
The case for $q=1$ was shown in \cite[Thm. 2]{Stojnic_Parvaresh_Hassibi_2009} and the case for $0<q<1$ was shown in \cite[Thm. 26]{latushkin_null_2015}.
\end{proof}
Using Lemma \ref{lemma:lq_recoverability}, we have the following result for the localization error identification problem in the case of $\bold R = \bold R_D$. 
\begin{theorem}
Suppose $\bold R_D$ has maximal rank and $d=2$. Then,
a solution $\bold x$ of the $l_2/l_q$ minimization problem P2, with $0<q\leq 1$ and $\bold R=\bold R_D$, is the unique solution to P2 if and only if
\begin{equation}
   \sum_{i\in\mathcal S}\| \bold p[i] - \bold c\|^q_{2} < \sum_{i\in\mathcal S^\complement}\| \bold p[i] - \bold c\|^q_{2}
    \label{eq:point_l_q_condition}
\end{equation}
for all $\bold c\in \mathbb R^2$ and for all subsets $\mathcal S\subseteq \{1, 2, \dots, |\mathcal V|\}$ with $|\mathcal S|\leq \|\bold x\|_{2,0}$.
\label{theorem:lq_recoverability}
\end{theorem}
\begin{proof} Using Lemma \ref{lemma:distance_rigidity_properties}, we see that any vector $\bold v\in\ker(\bold R_D)$ is a block vector having the form $\bold v[i] = \tilde{\bold c} + k\hspace{1pt}\bold S\hspace{1pt}\bold p [i]$,
where 
\[\bold S=\begin{bmatrix*}[r]
    0 & 1 \\ -1 & 0
\end{bmatrix*},
\]
$\tilde {\bold c} \in \mathbb R^2$, and $k\in\mathbb R$. Note that $\tilde {\bold c}$ and $k$ have the same value in each block of $\bold v$.
Using these facts in (\ref{eq:block_nsp_condition}) gives us the condition
\begin{equation}
    \left(\sum_{i\in \mathcal S}\|\tilde {\bold c} + k\hspace{1pt}\bold S\hspace{1pt}\bold p[i] \|^q_2\right)^{1/q} < \left(\sum_{i\in \mathcal S^\complement}\|\tilde {\bold c}+k\hspace{1pt}\bold S\hspace{1pt}\bold p[i]\|^q_2\right)^{1/q}
    \label{eq:thm_lq_before_c}
\end{equation}
Noting that $\bold S$ is an orthogonal matrix and dividing both sides by $k$, we have
\begin{equation}
   \left(\sum_{i\in \mathcal S}\|\tfrac{1}{k}\bold S^\top\tilde {\bold c} + \bold p[i] \|^q_2\right)^{1/q} < \left(\sum_{i\in \mathcal S^\complement}\|\tfrac{1}{k}\bold S^\top\tilde {\bold c} + \bold p[i]\|^q_2\right)^{1/q}
   \label{eq:thm_lq_after_c}
\end{equation}
Finally, we set $\bold c = - \tfrac{1}{k}\bold S^\top\tilde {\bold c}$ and use the operation $({}\cdot{})^q$ on both sides
to obtain (\ref{eq:point_l_q_condition}).
\end{proof}

To see the relevance of Theorem \ref{theorem:lq_recoverability}, consider $q=1$. Equation (\ref{eq:point_l_q_condition}) can be rewritten as
\begin{equation}
   \sum_{i\in\mathcal S}\| \bold p[i] - \bold c\|_{2} < \frac{1}{2}\sum_{i\in\mathcal V}\| \bold p[i] - \bold c\|_{2}
    \label{eq:point_l_1_condition}
\end{equation}
where for an arbitrary point $\bold c \in \mathbb R^2$, the left hand side is maximized by the set $\mathcal S$ corresponding to the positions of the $|\mathcal S|$ agents that are furthest from $\bold c$.
Thus, inequality (\ref{eq:point_l_1_condition}) is saying that the sum of distances of the $|\mathcal S|$ furthest agents from $\bold c$ accounts for less than half of the total sum of distances from $\bold c$, 
i.e., the agents in $\mathcal S^\complement$ are not clustered together near $\bold c$. 

As inequality (\ref{eq:point_l_1_condition}) must hold for all $\mathcal S\subseteq \mathcal V$ such that $|\mathcal S|\leq \|\bold x\|_{2,0}$, there is a maximal value for $\|\bold x\|_{2,0}$ such that inequality (\ref{eq:point_l_1_condition}) is satisfied by a given sensor network configuration, which is the maximum number of localization errors that can be uniquely recovered through $l_2/l_1$ minimization. We can generally expect this number to increase with the size of the network, since as the size of the sensor network increases, so does the right hand side of (\ref{eq:point_l_1_condition}). This indicates that up to a fraction of the agents in $\mathcal V$ can have localization errors before $l_2/l_q$ recovery fails to recover $\bold x$. A similar result can be shown for the case of $d=3$.
\begin{theorem}
The main claim of Theorem \ref{theorem:lq_recoverability} holds for $3$-dimensional sensor networks (i.e., $d=3$) if and only if inequality (\ref{eq:point_l_q_condition}) is satisfied by each of the $2$D configurations obtained by projecting $\{\bold p[i]\}_{i\in \mathcal V}$ onto an arbitrary plane in $\mathbb R^3$.
\label{theorem:3d_lq_recovery}
\end{theorem}
\begin{proof}
For $d=3$, a null vector $\bold v \in \ker(\bold R_D)$ has the form $\bold v[i] = \tilde{\bold c} + k\hspace{1pt} {\bold S}\hspace{1pt}\bold p [i]$, where $k\in \mathbb R$ and $ {\bold S}$ is a $3\times 3$ skew-symmetric matrix. We use the result from 
\cite[Sec. 7.3]{hogben2013handbook}
that any $3\times 3$ skew symmetric matrix $\bold S$ can be expressed as \mbox{$\bold S = l\hspace{1pt}\bold Q \bold \Lambda\bold Q^\top$},
where
\begin{equation}
  \bold \Lambda =  \begin{bmatrix*}[r]
0 & 1 & 0\\
-1 & 0 & 0\\
0 & 0 & 0
    \end{bmatrix*},
\end{equation}
$l\in \mathbb R$, and $\bold Q$ is an orthogonal matrix. Thus, $\bold v$ has the form $\bold v[i] = \tilde{\bold c} + kl\hspace{1pt} {\bold Q}{\bold \Lambda}{\bold Q}^\top\bold p [i]$.
Let $\tilde {\bold p}[i]=\bold Q^\top\bold p[i]$, such that $\tilde {\bold p}[i]$  denotes the agents' positions in a rotated coordinate frame. Following again the steps between (\ref{eq:thm_lq_before_c}) and (\ref{eq:thm_lq_after_c}), we have the condition
\begin{equation}
       \sum_{i\in \mathcal S}\|  \bold \Lambda \tilde{ \bold p}[i] - \bold c \|^q_2 < \sum_{i\in \mathcal S^\complement}\| \bold \Lambda \tilde{ \bold p}[i] - \bold c\|^q_2
\end{equation}
where $\bold c = - \frac{1}{kl}\bold Q^\top \tilde {\bold c}$. Observe that $\bold \Lambda$ projects (then rotates) the points onto a plane, as $
\bold \Lambda \tilde {\bold p}[i] = \left[\tilde{\bold p}[i]_y\ -\tilde{\bold p}[i]_x\  0\right]^\top$. To complete the argument, observe that the plane being projected onto is determined by $\bold Q$, which can be chosen arbitrarily.
\end{proof}

In the case of bearing-based localization error recovery, the conditions for $d=2$ and $3$ are similar to that of Theorem \ref{theorem:lq_recoverability}. This is related to the fact that the dimension of the bearing measurements (which are represented by unit vectors) equals the dimension of the space, $d$, as opposed to distance measurements which are scalar-valued irrespective of $d$.
\begin{proposition}
Suppose $\bold R_B$ has maximal rank and $d=2$. Then,
a solution $\bold x$ of the $l_2/l_q$ minimization problem P2, with $0<q\leq 1$ and $\bold R=\bold R_B$, is the unique solution to P2 if and only if
\begin{equation}
   \sum_{i\in\mathcal S}\| \bold p[i] - \bold c\|^q_{2} < \sum_{i\in\mathcal S^\complement}\| \bold p[i] - \bold c\|^q_{2}
    \label{eq:point_l_q_condition_Bearing}
\end{equation}
for all $\bold c\in \mathbb R^d$ and for all subsets $\mathcal S\subseteq \{1, 2, \dots, |\mathcal V|\}$ with $|\mathcal S|\leq \|\bold x\|_{2,0}$.
\label{thm:l_q_recoverability_Bearing}
\end{proposition}
\begin{proof}
Observe that a vector $\bold v \in \ker (\bold R_B)$ has the form $\bold v[i] = \tilde {\bold c} + k \hspace{1pt}\bold p [i]$, where $\tilde {\bold c} \in \mathbb R^d$ and $k \in \mathbb R$.
Thereafter, the proof follows similarly to those of Theorems \ref{theorem:lq_recoverability} and \ref{theorem:3d_lq_recovery}.
\end{proof}

Thus, following a reasoning similar to the one we used for distance-based localization error recovery, we observe that $l_2/l_1$ minimization can be used to recover the localization errors by processing the inter-agent bearing measurements as well, as long as the number of localization errors is not too large (i.e., beyond the maximal value of $|\mathcal S|$ for which inequality (\ref{eq:point_l_q_condition_Bearing}) holds).

%% file: sections/5_robustness.tex
\section{Robustness Guarantees}%
\label{sec:robustness}
In this section, we discuss the robustness of the $l_2/l_q$ minimization method subject to practical considerations such as measurement noise and linearization error. 
\subsection{Measurement Noise}
Suppose $\bold z = \bold R\bold x+\bold e$, where the measurement noise vector $\bold e$ is bounded as $\|\bold e\|_2\leq \epsilon$ with $\epsilon > 0$, then we are interested in solving the following robust $l_2/l_q$ minimization problem:
\begin{align*}
\begin{array}{rclc}
   \textsc{P3:}  &  &
\begin{array}{rl}
\underset{\hat{\bold x}}{\textnormal{minimize}}\quad  
&\| \hat{\bold x}\|_{2,q}  \\
\textnormal{subject to} \quad
& \|\bold z - \bold R \hat{\bold x}\|_2\leq \epsilon
\end{array} & \qquad \qquad
\end{array}
\end{align*}
where $\epsilon$ can be interpreted as the slackness of the constraint, which is introduced to accommodate the measurement noise.
In the case of $q=1$, the problem P3 is related to the basis pursuit denoising problem \cite{gill2011crowd_bpdn}.
Denote the solution of P3 as $\bold x^*$ which (due to the measurement noise) may be different from the vector $ {\bold x}$ that generated the observed measurements.
Nevertheless, we have the following relationship between $\bold x$ and $\bold x^*$.
\begin{lemma}[Robust $l_2/l_q$ Recoverability \protect{\cite[Cor. 3]{robust_NSP_2017}}]
Let $s>0$.
Suppose for all sets $\mathcal S\subseteq \{1, 2, \dots, n\}$ with $|\mathcal S|\leq s$, there exist constants $0<\tau < 1$ and $\gamma >0$ such that
\begin{equation}
\|\bold v_{\mathcal S}\|_{2,q} \leq \tau\|\bold v_{\mathcal S^\complement}\|_{2,q} + \gamma \| \bold R \bold v\|_2
\label{eq:robust_lq_cond}
\end{equation}
for all $\bold v\in \mathbb R^{d|\mathcal V|}$. Then, a solution $\bold x^*$ of Problem P3 approximates $ {\bold x}$ with the error
\begin{equation}
    \| \bold x - \bold x^*\|_{2,q} \leq C_{q, \tau}\sigma _{q,s}( \bold x) + C_{q,\tau,\gamma}\epsilon,
    \label{eq:robust_lq_ineq}
\end{equation}
\begin{align}
&\textit{where} \nonumber \\
& C_{q, \tau}=2^{\sfrac{2}{q}-1}\left(\frac{1+\tau^q}{1-\tau^q}\right)^{\sfrac{1}{q}};\   C_{q,\tau,\gamma}=2^{\sfrac{2}{q}}\left(\frac{1}{1-\tau^q}\right)^{\sfrac{1}{q}}\gamma\ ;\nonumber\\
&
\textit{and}\quad       \sigma_{q,s}(\bold x) = \inf \left\lbrace  \|\bold x - \bold z\|_{2,q} 
\hspace{1pt}\big\vert\hspace{1pt}
\bold z \in \mathbb R^{d|\mathcal V|}, \|\bold z\|_{2,0}\leq s  \right\rbrace.
\nonumber
\end{align}
\label{lem:robust_lq}
\end{lemma}
Using Lemma \ref{lem:robust_lq}, we can show that $l_2/l_1$ recoverability in the noise-free case (which was discussed in the previous section) also guarantees a certain degree of robustness in the presence of noise.
\begin{theorem}
Suppose that $\bold R$ satisfies the $l_1$ block NSP of order $s$. 
Then, there exist constants $0<\tau<1$ and $0<\gamma<\infty$ such that any solution $\bold x^*$ of the robust $l_2/l_q$ minimization problem P3, with $q=1$, is related to $\bold x$ via inequality (\ref{eq:robust_lq_ineq}).
\label{thm:l1_robust}
\end{theorem}
\begin{proof}
The proof is given in the appendix.
\end{proof}
\noindent
A crucial observation from this proof is that the constant $\gamma$ in (\ref{eq:robust_lq_cond}) (which represents the sensitivity of the approximation error to measurement noise) is inversely proportional to the smallest non-zero eigenvalue of $\bold R^T \bold R$, which is denoted as $\tilde \lambda(\bold R^T\bold R)$. In the literature, $\tilde \lambda(\bold R^T\bold R)$ is called as the \textit{worst-case rigidity index} and is considered as a quantitative measure of the rigidity of the configuration \cite{zhu2009stiffness, trinh2016further, jordan2022rigidity}. By using distributed control algorithms to maximize the worst-case rigidity index of a sensor network configuration, the robustness of the  $l_2/l_1$ minimization approach to measurement noise can be enhanced \cite{decent_dist_rigidity_2015, bearing_rigidity_maintenance_2017}.

 \subsection{Imperfect Position Estimates}
In this subsection, we address the case where the agents in $\mathcal D^\complement$ have obtained imperfect estimates of their own positions due to the uncertainty in their dynamical and measurement models. In this case,
we have $\mathbf x [i]\neq{\bold 0}\ \forall i \in \mathcal V$, i.e., $\bold x$ is not a block-sparse vector. As each agent in $\mathcal V$ has a localization error, we consider the problem of identifying the \textit{significant localization errors} instead which are defined as follows. Given $\kappa > 0$, agent $i\in \mathcal V$ is said to have a significant localization error if
$\|\bold x[i]\|_2 > \kappa$. Let ${\mathcal D} \subseteq \mathcal V$ be the set of agents which have significant localization errors.

While the noise in the inter-agent measurements contributes to the second term on the right hand side of inequality (\ref{eq:robust_lq_ineq}), the error in the position estimates of the agents in $ {\mathcal D}^\complement$, $\bold x_{ {\mathcal D}^\complement}$, contributes to the term $\sigma_{q,s}(\bold x)$. Suppose the sensor network configuration satisfies the $l_1$ block NSP of order $s$, such that $s \geq | {\mathcal D}|$, then we have the following upper bound on $\sigma_{1,s}(\bold x)$:

\begin{equation}
\sigma_{1,s}(\bold x) \leq 
\big(|\mathcal V|-s\big) \kappa
\end{equation}
which follows from the observation that the infimum in the expression for $\sigma_{q,s}(\bold x)$ is attained only if $\bold z[i]=\bold x[i] \Leftrightarrow i\in  {\mathcal D}$. Thus, the approximation error increases linearly with respect to the number of agents in the sensor network.


Interestingly, the presence of randomness in the entries of $\bold p$ and $\bold R$ is beneficial for the theoretical analysis of rigidity;
a random set of positions, $\{\bold p[i]\}_{i\in\mathcal V}$, is said to be \textit{generic}, which means that there is no algebraic dependence between the positions of the agents. Generic configurations play an important role in rigidity theory as it is easier to ensure rigidity in generic configurations (see \cite{Eren_rigidity_randomness_2004,bearing_rigidity_maintenance_2017}).
Moreover, with probability one, no more than $2$ points in $\{\bold p[i]\}_{i\in\mathcal V}$ can be colinear in a random configuration. Thus, a stronger guarantee of $l_2/l_0$ recoverability can be established for random configurations using Theorem \ref{theorem:l0_recoverability}. This may be compared to the role of random matrices in compressive sensing theory, which are used to guarantee recoverability properties with probability one \cite{Eldar_Mishali_2009}.

\subsection{Linearization Error}

In practice, the linearization error introduced by using $\bold R(\hat {\bold p})$ in place of the nonlinear measurement model can undermine the performance of our approach. This problem can be addressed by using an iterative approach to solve the robust $l_2/l_q$ minimization problem P3, using the idea of \textit{bootstrapping}.
Suppose $\bold x_k$ is the estimate of $\bold x$ at a given iteration of the convex optimization routine, we can use $\bold R(\hat {\bold p} + \bold x_k)$ in place of $\bold R(\bold p)$ as a bootstrapped estimate of the rigidity matrix, based on the observation that $\hat {\bold p} + \bold x_k$ is closer to $\bold p$ than $\hat {\bold p}$ is.

With this motivation, we propose an algorithm based on sequential convex programming (SCP) (given in Algorithm \ref{alg:scp}) to iteratively reconstruct $\bold x$, following the methodology of \cite{zillober2004very}. In each iteration, the SCP algorithm updates its estimate of the rigidity matrix, and then
solves the corresponding robust $l_2/l_1$ minimization problem using a convex optimization solver. The stopping condition for the SCP algorithm, given in step $5$ of Algorithm \ref{alg:scp}, is based on the following proposition. It shows that in a sensor network where $\bold \Phi(\hat{\bold p}) \neq \bold \Phi({\bold p})$ (i.e., the inter-agent measurement residuals are non-zero) due to localization errors, Algorithm \ref{alg:scp} recovers a sparse vector $\bold x^*$ which serves as an \textit{explanation} for the observed measurement residuals.

\begin{proposition}
Consider the noise-free case (i.e., $\bold e=\bold 0$) of the localization error recovery problem. Suppose in Algorithm $\ref{alg:scp}$, we have $\|\tilde {\bold x}_k\|_{2} \rightarrow 0$ as $\epsilon_{k} \rightarrow 0$, then
\begin{equation}
    \lim_{\epsilon_k \rightarrow 0} \big(\bold \Phi (\hat {\bold p} + \bold x_k)\big) = \bold \Phi (\bold p)
    \label{eq:local_convergence}
\end{equation}
\end{proposition}
\begin{proof}
Since $\|\tilde {\bold x}_k\|_{2} \rightarrow 0$, there exists some $\bold x^*\in \mathbb R^{d|\mathcal V|}$ such that $\lim_{\epsilon_k \rightarrow 0}({\bold x}_k) = \bold x^*$, giving us
\begin{equation}
\lim_{\epsilon_k \rightarrow 0} \left(\bold z_{k}\right) = \bold y- \bold \Phi(\hat{\bold p}+\bold x^*)
\label{eq:z_limit}
\end{equation}
Moreover, as $\tilde {\bold x}_k$ is a solution to the optimization problem in step $4$ of Algorithm \ref{alg:scp}, we have
\begin{equation}
    \lim_{\epsilon_k\rightarrow 0}\left( \bold z_{k} -  \bold R_{k} \tilde {\bold x}_k\right) = \bold 0 
    \label{eq:final_prop_limiteq}
    \end{equation}
Using the definition of $\bold z_k$ and substituting $\bold y$ with $\bold \Phi(\bold p)$, we can rewrite (\ref{eq:final_prop_limiteq}) as
    \begin{equation}
    \lim_{\epsilon_k \rightarrow 0} \Big( \bold \Phi (\bold p) - \bold \Phi (\hat {\bold p} + \bold x_k) - \nabla \bold \Phi (\hat {\bold p} + \bold x_k)\tilde {\bold x}_k \Big)= \bold 0
    \label{eq:prop_local_convergence_lasteqn}
\end{equation}
Since each of the summands has a limit (due to the continuity of $\bold \Phi$ and $\nabla\bold \Phi$), we can use the fact that $\tilde {\bold x}_k \rightarrow \bold 0$ in (\ref{eq:prop_local_convergence_lasteqn}) to obtain (\ref{eq:local_convergence}).
\end{proof}

Thus, Algorithm \ref{alg:scp} is able to uniquely recover both the set of incorrectly localized agents as well as their corrected positions. We use Algorithm \ref{alg:scp} to numerically demonstrate the proposed approach in the next section, showing that it recovers $\bold x$ accurately in as few as $4$ SCP iterations.

\begin{algorithm}
\caption{Sparse Localization Error Recovery using SCP}
\begin{algorithmic}[1]
\vspace{2pt}
\REQUIRE 
 $\hat{\bold p}$, $\bold y$, $N_{\text{Iterations}}\geq 1$, initial constraint slackness $\epsilon_0\geq 0$, slackness reduction factor $0<\rho\leq1$, and tolerance $\delta \geq 0$.\\
\vspace{3pt} 
\STATE $\bold x_0 \leftarrow \bold 0$
\FOR{$k=1,\dots,N_{\text{Iterations}}$}
\STATE Construct the residual vector and rigidity matrix: 
\[\bold z_{k-1} = \bold y - \bold \Phi (\hat {\bold p} + \bold x_{k-1})\] 
\[\bold R_{k-1} = \nabla \bold \Phi(\hat {\bold p} + \bold x_{k-1})\]
\STATE Update the estimated error vector using \[\bold x_k \leftarrow \bold x_{k-1} + \tilde {\bold x}_{k-1},\]
where $\tilde {\bold x}_{k-1}$ is the solution to the following convex optimization problem:
\begin{align*}
\begin{array}{rl}
\underset{\bold x}{\textnormal{minimise}}\quad  
&\|\bold x_{k-1} + \bold x\|_{2,1}  \\
\textnormal{subject to} \quad
& \|\bold z_{k-1} - \bold R_{k-1}\bold x\|_2 \leq \epsilon_{k-1}
\end{array}
\end{align*}
\IF {$\|\tilde {\bold x}_{k-1}\|_2<\delta$}
\STATE \textbf{break}
\ELSE 
\STATE Tighten the constraint: $\epsilon_k \leftarrow \rho\hspace{1pt}\epsilon_{k-1}$
\ENDIF
\ENDFOR
\STATE $\bold x^* \leftarrow \bold x_k$
\RETURN The recovered localization error vector, $\bold x^*$.
\end{algorithmic}
\label{alg:scp}
\end{algorithm}

%% file: sections/6_simulation.tex
\section{Numerical Example}
\label{sec:numerical}

In this section, we consider a numerical example of a sensor network consisting of $13$ uncrewed aerial vehicles (UAVs) positioned in a $3$-dimensional configuration. It is assumed that the UAVs are able to measure their distances from the others in their vicinity by using ultra wideband (UWB) ranging \cite{cao2021relative}. The SCP algorithm (given in Algorithm \ref{alg:scp}) is used to process the distance measurements, identify the UAVs having localization errors, as well as recover their true positions.

The details of the numerical example are as follows; the network of UAVs is represented by a graph $\mathcal G=(\mathcal V, \mathcal E)$ having $13$ vertices and $36$ edges, where the vertices represent the UAVs and the edges correspond to the availability of distance measurements between them. The sensor network's configuration is determined by the positions of its vertices in $\mathbb R^3$ (shown in Fig. \ref{fig:single_run}), which are collectively represented by the block vector $\bold p$. Each UAV uses a suite of onboard sensors such as inertial navigation systems (INS), visual odometry, and GNSS receivers to estimate its own position, such that the position estimates of the UAVs are collectively represented by the block vector, $\hat{\bold p}$. 
Due to GNSS multipath errors and/or spoofing of GNSS signals, a subset of the UAVs, $\mathcal D\subseteq \mathcal V$, have localization errors.

Letting $\bold R_D$ denote the distance rigidity matrix of the sensor network, we have that the $7^{th}$-smallest eigenvalue of $\bold R_D^\top \bold R_D$ is equal to $0.334$, whereas $6$ of its eigenvalues are smaller than $10^{-14}$ (i.e., are close to zero). Using Lemma \ref{lemma:distance_rigidity_properties}, we know that the rank of $\bold R_D$ is maximal, and thus, the sensor network is infinitesimally rigid.
Using Theorem \ref{theorem:3d_lq_recovery}, we compute that $l_2/l_1$ minimization can be used to uniquely recover the localization error vector, $\bold x=\bold p - \hat{\bold p}$, as long as the number of localization errors is no greater than $6$. In the presence of linearization error (which was not considered in Theorem \ref{theorem:3d_lq_recovery}), the number of localization errors which can be recovered with certainty turns out to be about $4$, as demonstrated in Section \ref{subsection:n_faulty}.

\subsection{Noise-Free Case}
\label{subsection:sim_noisefree}
We first consider the noise-free case, where the measurement noise vector $\bold e$ is set to $\bold 0$, and the agents without localization errors are assumed to have perfect estimates, i.e., $\hat {\bold p}_{\mathcal D^\complement} = \bold p_{\mathcal D^\complement}$. The set of localization errors, $\mathcal D$, is constructed by randomly selecting $4$ agents from $\mathcal V$. For each agent in $\mathcal D$, its localization error is chosen by generating independent samples from a random vector having the uniform distribution over the $3$-dimensional unit cube. 

Figure \ref{fig:single_run} shows the non-zero blocks of the vector $\bold x^*$ recovered by the SCP algorithm (Algorithm \ref{alg:scp}) after $4$ iterations, along with the blocks of $\bold p$ and $\hat{\bold p}$ (which represent the true and estimated positions of each of the agents, respectively), demonstrating that the proposed $l_2/l_1$ minimization approach is able to uniquely recover the localization error vector as expected. For the SCP algorithm, its parameters $\epsilon_0$ and $\rho$ were chosen as $4.0$ and $3.0$, respectively, based on an iterative parameter tuning process.
Figure \ref{fig:error_vs_scp_iterations} shows the relative approximation error, defined as $\|\bold x - \bold x^*\|_2 / \| \bold x \|_2$, as a function of the number of SCP iterations, showing that a larger number of SCP iterations can be used to improve the accuracy of the error recovery.

\begin{figure}[!t]
\centerline{\includegraphics[trim={0.7cm, 1.5cm, 0.8cm, 2.0cm},clip, width=0.9\columnwidth]{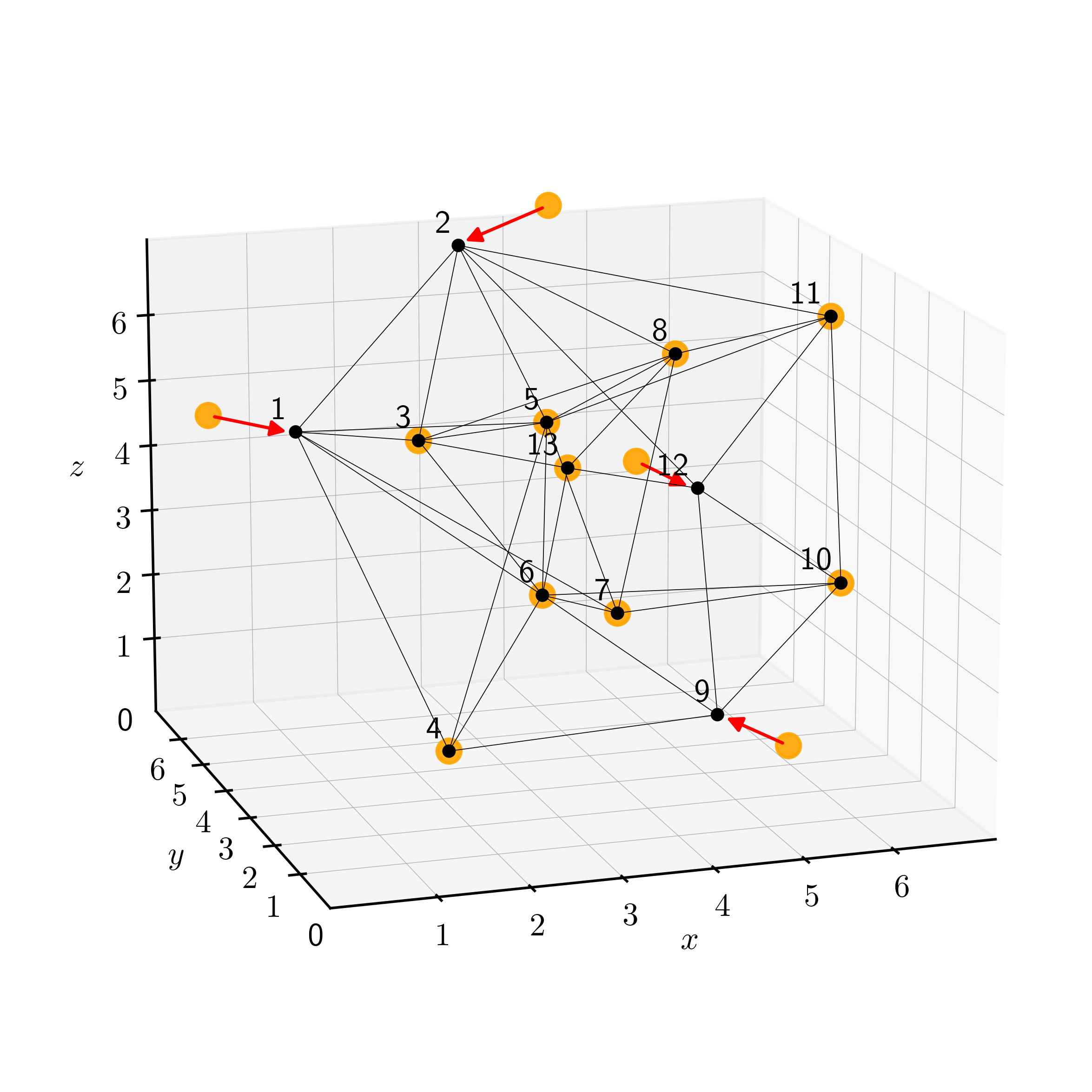}}
\caption{The true positions of the UAVs, $\mathbf p[i]$, are depicted by the black points and their estimated positions, $\hat {\mathbf p}[i]$, are depicted by the yellow discs. The localization error vector $\bold x^*$ was recovered using $4$ iterations of the SCP algorithm; the non-zero blocks of $\bold x^*$ are depicted by the red arrows.}
\label{fig:single_run}
\end{figure}

\begin{figure}[!t]
\centerline{\includegraphics[trim={0.0cm, 0.1cm, 0.0cm, 0.1cm},clip, width=0.92\columnwidth]{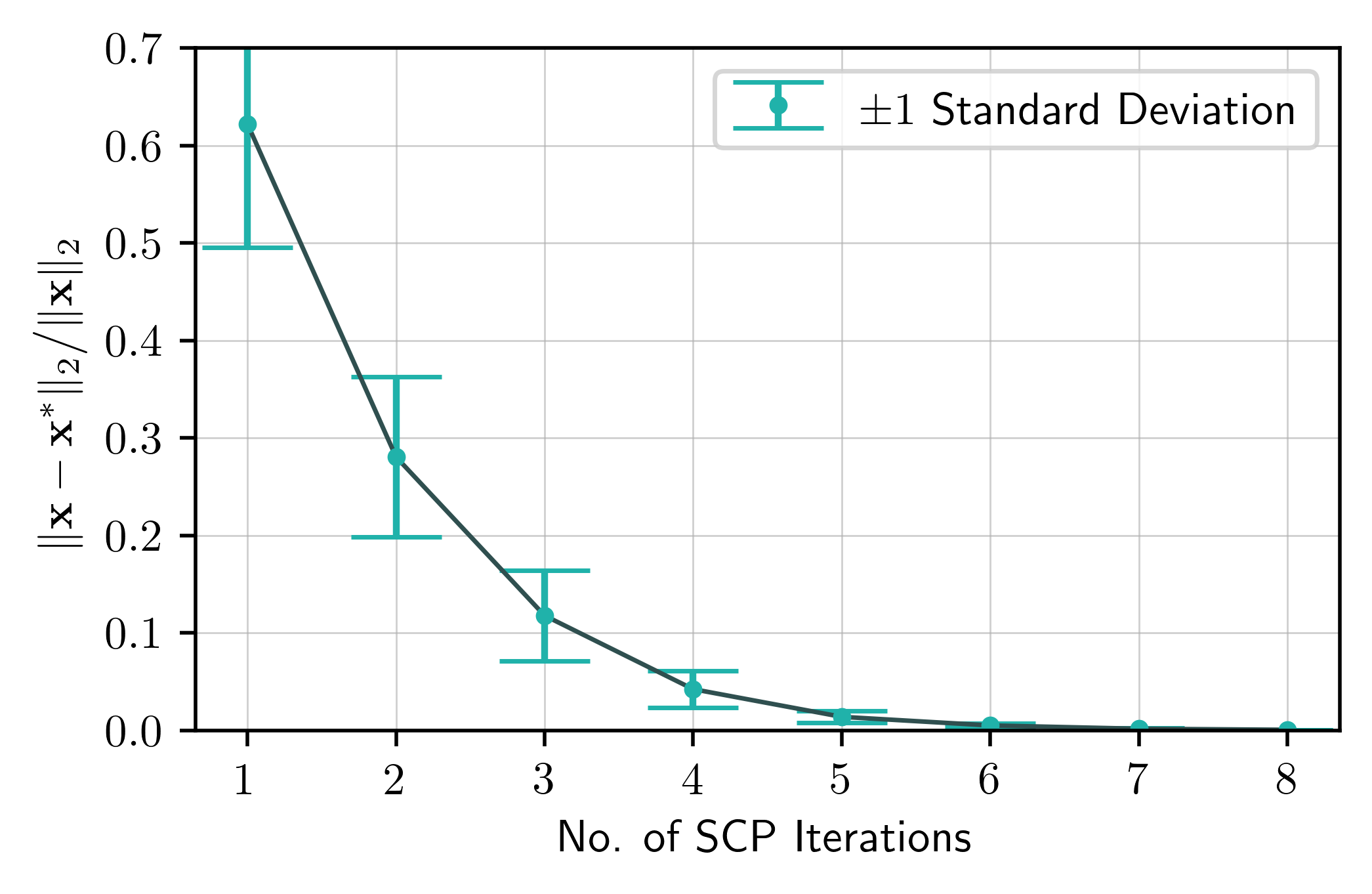}}
\caption{The relative approximation error, $\|\bold x - \bold x^*\|_2 / \| \bold x \|_2$, as a function of the number of SCP iterations used to recover the localization errors. The standard deviation of the error was computed numerically by simulating $250$ Monte Carlo trials.}
\label{fig:error_vs_scp_iterations}
\end{figure}

\subsection{Uncorrelated vs. Fully Correlated Errors}
\label{subsection:n_faulty}

So far, the localization errors (i.e., the non-zero blocks of $\bold x$) were sampled in an independent and identically distributed (i.i.d.) manner, such that the localization errors are \textit{uncorrelated}, i.e., $\mathbb E\left[\bold x[i] \bold x[j]^\top\right]=\bold 0$ whenever $i\neq j$. The other extreme case is that of \textit{fully correlated} localization errors, i.e., $\bold x[i]=\bold x[j],\ \forall i,j \in \mathcal D$. For either type of localization errors, we repeat the simulation scenario considered in Section \ref{subsection:sim_noisefree} for $250$ Monte Carlo trials. In each simulation, the identities of the agents in $\mathcal D$ were selected at random.

Given the block vector $\bold x^*$ obtained using the SCP algorithm, let $\hat {\mathcal D}$ denote the locations of its non-zero blocks, such that $\|\bold x[i]\|_2 > 0 \Leftrightarrow i \in \hat {\mathcal D}$.
Figure \ref{fig:error_vs_n_faulty} shows the percentage of Monte Carlo trials in which the SCP algorithm was able to identify the localization error vector correctly (i.e., $\hat {\mathcal D} = \mathcal D$) within $4$ iterations. It can be seen that the SCP algorithm is able to recover a smaller number of localization errors when the errors are fully correlated. In either case, up to $4$ localization errors were identified with a high degree of certainty.

\begin{figure}[!t]
\centerline{\includegraphics[trim={0.0cm, 0.1cm, 0.0cm, 0.1cm},clip, width=0.92\columnwidth]{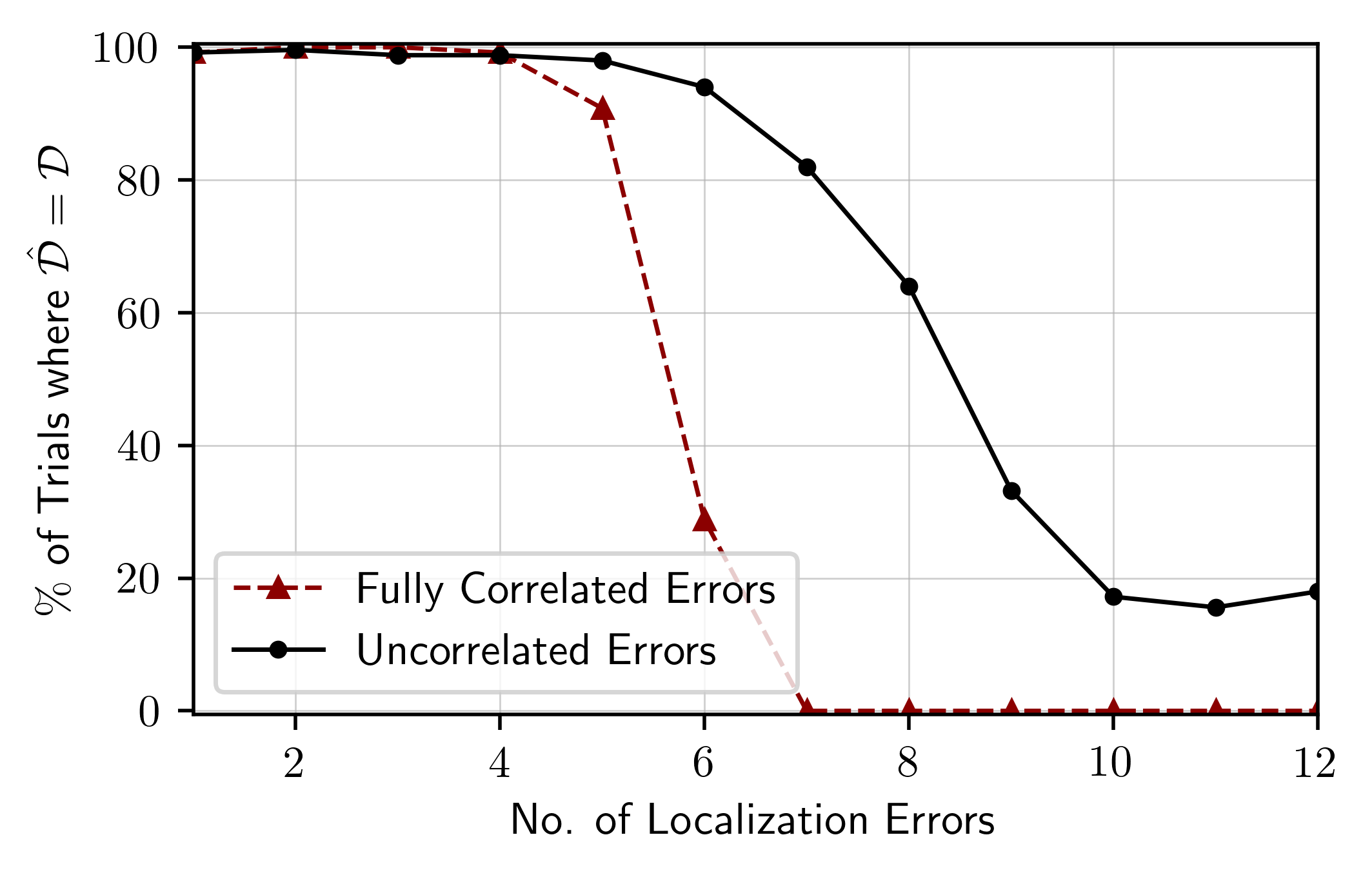}}
\caption{The vertical axis corresponds to the percentage of Monte Carlo trials (out of $250$ trials) in which the SCP algorithm was able to correctly identify the set of UAVs which have localization errors, i.e., $\hat {\mathcal D} = \mathcal D$.}
\label{fig:error_vs_n_faulty}
\end{figure}

\subsection{Robustness}
Finally, we assess the robustness of the $l_2/l_1$ minimization approach subject to measurement noise and imperfect position estimates. We consider the simulation scenario described in Section \ref{subsection:sim_noisefree}, wherein $4$ of the UAVs (selected at random) are subject to independent random localization errors. As discussed in Section \ref{sec:robustness}, the magnitude of measurement noise is dictated by the quantity $\epsilon$. We consider the following values of $\epsilon$: $0, 1, 2, 3, 4,$ and $5$, for which the corresponding slackness reduction parameter $\rho$ of the SCP algorithm is set to $3.0, 2.0, 1.5, 1.5, 1.3,$ and $1.2$, respectively. The motivation for choosing the foregoing values of $\rho$ is that when the magnitude of measurement noise is high, a greater amount of slack is required in the constraint of problem P3 in order to make it feasible. Similarly, the initial constraint slackness $\epsilon_0$ of the SCP algorithm is set to $4 + \epsilon$, wherein the first term accommodates the $4$ localization errors, whereas the second term accommodates the measurement noise. For each simulation, we vary the value of $\kappa$ (which dictates the amount of imperfection in the position estimates of the agents in $\mathcal D^\complement$) between $0.0$ to $0.9$.

The measurement noise vector is sampled from the uniform random distribution over the surface of a sphere centred at the origin, having the radius $\epsilon$. Similarly, the position estimates are generated such that for each $i \in \mathcal D^\complement$, $\hat {\bold p}[i]$ is sampled from the uniform random distribution over the surface of a sphere centred at $\bold p[i]$, having the radius $\kappa$. Thus, both the measurement noise and the imperfection in the position estimates are random, whereas their corresponding magnitudes are chosen deterministically as $\epsilon$ and $\kappa$.
This allows us to visualize the dependency of the approximation error on $\epsilon$ and $\kappa$, which is shown in Fig. \ref{fig:error_vs_measurement_error}. 
Observe that the plots in Fig. \ref{fig:error_vs_measurement_error} are bounded from above by $1.0$, indicating that while the performance of $l_2/l_1$ minimization degrades as $\epsilon$ and/or $\kappa$ increase, the $l_2/l_1$ minimization approach can partially recover the localization error vector $\bold x$ up to some approximation error. 

\begin{figure}[!t]
\centerline{\includegraphics[trim={0.0cm, 0.1cm, 0.0cm, 0.1cm},clip, width=0.92\columnwidth]{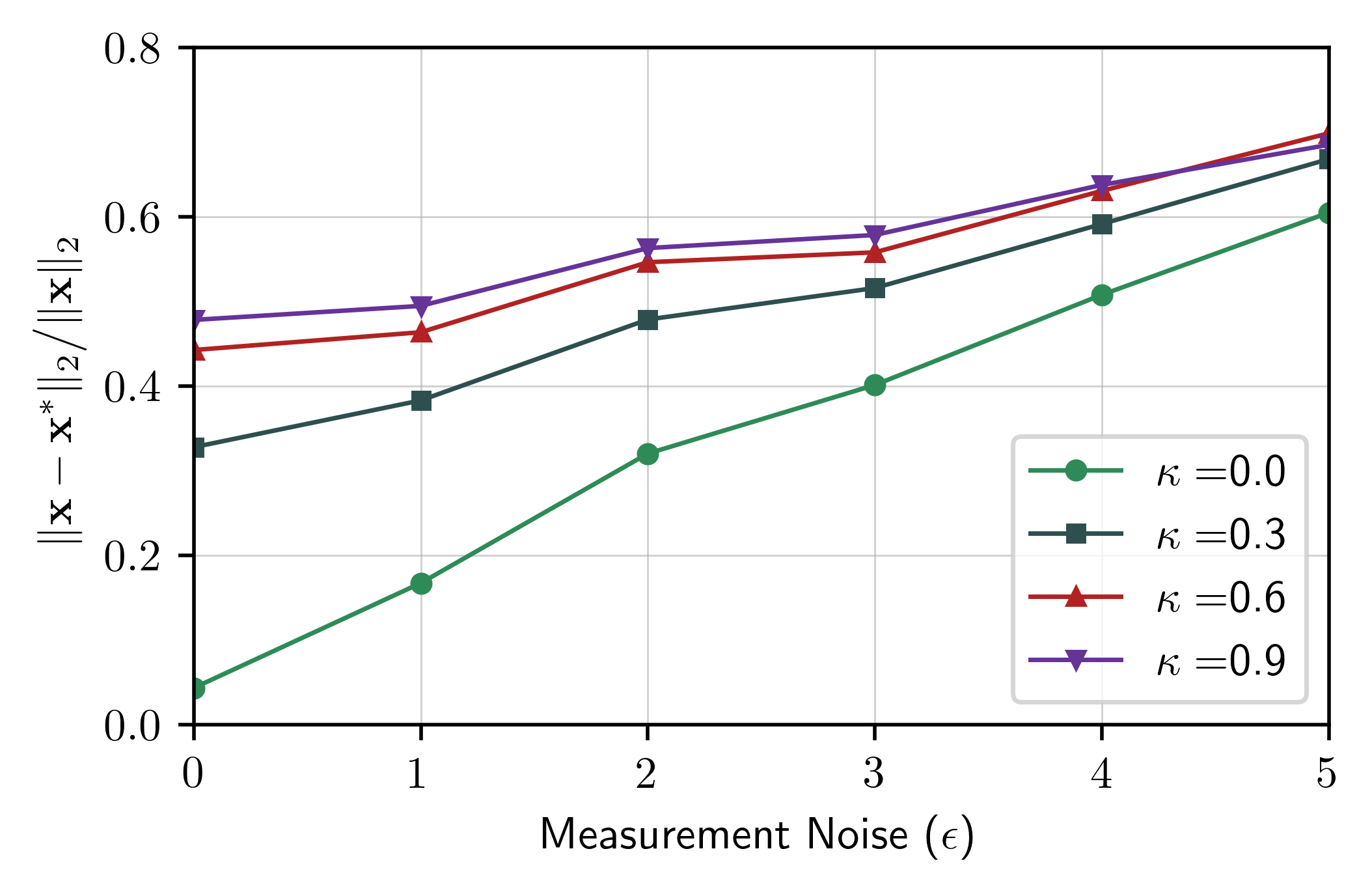}}
\caption{The relative approximation error, $\|\bold x - \bold x^*\|_2 / \| \bold x \|_2$, as a function of the magnitude of measurement noise ($\epsilon$) and the imperfection in the position estimates ($\kappa$), averaged over $250$ Monte Carlo trials.}
\label{fig:error_vs_measurement_error}
\end{figure}

\section{Conclusions}
In this paper, we developed a novel approach for recovering (i.e., identifying and reconstructing) the localization errors in a sensor network by processing inter-agent measurements such as distances and/or bearings. The proposed method brings together analytical tools from compressive sensing and rigidity theory, allowing us to establish necessary and sufficient conditions on the sensor network's configuration and connectivity which guarantee that a given number of localization errors can be recovered uniquely. The proposed $l_2/l_1$ minimization method for recovering the localization errors is an efficient convex optimization problem that circumvents the combinatorial complexity of the approaches found in the literature. Moreover, the proposed method does not require a subset of the agents of the sensor network to be labeled as the \textit{anchors}, which differentiates our method from each of the existing approaches for sensor network localization using inter-agent measurements. Using a numerical example, we demonstrated our findings and benchmarked the performance and robustness properties of the proposed $l_2/l_1$ minimization approach for detecting localization errors.

Future work on this topic will focus on developing a distributed algorithm for solving the $l_2/l_1$ minimization problem, which can be used to recover the localization errors in large-scale sensor networks consisting of autonomous agents.

%% file: sections/7_appendices.tex
\appendix[Proof of Theorem \ref{thm:l1_robust}]
\begin{proof}
Given a vector $\bold v \in \mathbb R^{d|\mathcal V|}$, it can be uniquely decomposed as $\bold v = \bar{\bold v} + \bold{\tilde v}$, where $\bar{\bold v}\in \ker(\bold R)$ and $\bold{\tilde v}\in \ker(\bold R)^\perp$ are orthogonal. Thus, inequality (\ref{eq:robust_lq_cond}) can be rewritten as
\begin{equation}
\|\bold{\tilde v}_{\mathcal S} + \bar{\bold v}_{\mathcal S}\|_{2,1} - \tau\|\bold{\tilde v}_{\mathcal S^\complement}+\bar{\bold v}_{\mathcal S^\complement}\|_{2,1} \leq  \gamma \| \bold R \bold{\tilde v}\|_2
\label{eq:thm_robust_eq_1}
\end{equation}
Using the definition of $l_1$ block NSP (\ref{eq:block_nsp_condition}), we know that there exists some $\bar \tau$ such that $0<\bar \tau<1$ and for all subsets $\mathcal S\subseteq \{1, 2, \dots, |\mathcal V|\}$ with $|\mathcal S|\leq s$, we have
\begin{equation}
\|\bar{\bold v}_{\mathcal S}\|_{2,1} \leq \bar \tau \|\bar{\bold v}_{\mathcal S^\complement}\|_{2,1}
\label{eq:thm_robust_block_nsp_condition}
\end{equation}
To prove the theorem, we need to show that if the $l_1$ block NSP property (\ref{eq:thm_robust_block_nsp_condition}) holds, then so does (\ref{eq:thm_robust_eq_1}).

Observe that if $\bold{\tilde v}=\bold 0$, then (\ref{eq:thm_robust_block_nsp_condition}) implies (\ref{eq:thm_robust_eq_1}) if $\tau$ is chosen such that $\bar \tau \leq  \tau<1$. Hereafter, we consider the case where $\bold{\tilde v}\neq \bold 0$; it will be shown that (\ref{eq:thm_robust_block_nsp_condition}) implies (\ref{eq:thm_robust_eq_1})
if $\tau$ and $\gamma$ are chosen such that $\bar \tau \leq  \tau<1$ and 
\begin{equation}
\gamma  = (1+\bar \tau + \tau)
\sqrt{\frac{ |\mathcal V|}{{\tilde\lambda(\bold R^\top \bold R)}
}
}
\label{eq:gamma_def}
\end{equation}
where ${\tilde \lambda}(\bold R^\top \bold R)$ denotes the smallest non-zero eigenvalue of $\bold R^\top \bold R$. Rewriting (\ref{eq:thm_robust_eq_1}), we have
\begin{align}
 \frac{\|\bold{\tilde v}_{\mathcal S} + \bar{\bold v}_{\mathcal S}\|_{2,1} - \tau\|\bold{\tilde v}_{\mathcal S^\complement}+\bar{\bold v}_{\mathcal S^\complement}\|_{2,1}}
{\| \bold R \bold{\tilde v}\|_2}  \qquad \qquad \qquad \qquad &\nonumber
\\
\leq 
\frac{\|\bold{\tilde v}_{\mathcal S} + \bar{\bold v}_{\mathcal S}\|_{2,1} - \tau\|\bold{\tilde v}_{\mathcal S^\complement}+\bar{\bold v}_{\mathcal S^\complement}\|_{2,1}}{
{\tilde\lambda(\bold R^\top \bold R)}^{\sfrac{1}{2}}
\|\bold{\tilde v}\|_2}
&
\\
\leq 
\frac{\|\bold{\tilde v}_{\mathcal S} + \bar{\bold v}_{\mathcal S}\|_{2,1} - \tau\|\bold{\tilde v}_{\mathcal S^\complement}+\bar{\bold v}_{\mathcal S^\complement}\|_{2,1}}{
{\tilde\lambda(\bold R^\top \bold R)}^{\sfrac{1}{2}}{|\mathcal V|}^{-\sfrac{1}{2}}
\|\bold{\tilde v}\|_{2,1}}&
\label{eq:thm_robust_eq_2}
\end{align}
where in the last step, we used the equivalence of the norms $\|{}\cdot{}\|_{2}$ and $\|{}\cdot{}\|_{2,1}$ \cite[Eq. 15]{robust_NSP_2017}.
Focusing on the numerator of (\ref{eq:thm_robust_eq_2}), and using the fact that $\|{}\cdot{}\|_{2,q}$ is a norm for $q\geq1$ \cite{robust_NSP_2017} as well as (\ref{eq:thm_robust_block_nsp_condition}), we have
\begin{align}
\|\bold{\tilde v}_{\mathcal S} + &\bar{\bold v}_{\mathcal S}\|_{2,1} - \tau\|\bold{\tilde v}_{\mathcal S^\complement}+\bar{\bold v}_{\mathcal S^\complement}\|_{2,1}\nonumber\\
&\leq \|\bold{\tilde v}_{\mathcal S}\|_{2,1} +\|\bar{\bold v}_{\mathcal S}\|_{2,1} - \tau\|\bold{\tilde v}_{\mathcal S^\complement}+\bar{\bold v}_{\mathcal S^\complement}\|_{2,1}
 \quad\\
&\leq \|\bold{\tilde v}_{\mathcal S}\|_{2,1} + \bar \tau \|\bar{\bold v}_{\mathcal S^\complement}\|_{2,1} - \tau\|\bold{\tilde v}_{\mathcal S^\complement}+\bar{\bold v}_{\mathcal S^\complement}\|_{2,1}
\\
&\leq \|\bold{\tilde v}_{\mathcal S}\|_{2,1} + \bar \tau \|\bar{\bold v}_{\mathcal S^\complement}\|_{2,1} - \tau\Big|\|\bold{\tilde v}_{\mathcal S^\complement}\|_{2,1} - \|\bar{\bold v}_{\mathcal S^\complement}\|_{2,1}\Big|
\label{eq:thm_robust_eq_3}
\end{align}
where in the last step we used the reverse triangle inequality for a vector norm $\|{}\cdot{}\|$, which establishes for vectors $\bold v_1$ and $\bold v_2$ that $\|\bold v_1 + \bold v_2\|\geq\big|\|\bold v_1\| - \|\bold v_2\|\big|$ \cite[Sec. 50.1]{hogben2013handbook}.
\vspace{4pt}
\begin{itemize}
    \item \textit{Case 1:} If $\|\bold{\tilde v}_{\mathcal S^\complement}\|_{2,1} \geq \|\bar{\bold v}_{\mathcal S^\complement}\|_{2,1}$, then (\ref{eq:thm_robust_eq_3}) simplifies to
\begin{align}
\|&\bold{\tilde v}_{\mathcal S}\|_{2,1} + \bar \tau \|\bar{\bold v}_{\mathcal S^\complement}\|_{2,1} - \tau\|\bold{\tilde v}_{\mathcal S^\complement}\|_{2,1} + \tau\|\bar{\bold v}_{\mathcal S^\complement}\|_{2,1}\\
&= \|\bold{\tilde v}\|_{2,1} - (1+\tau)\|\bold{\tilde v}_{\mathcal S^\complement}\|_{2,1} + (\bar \tau + \tau)\|\bar{\bold v}_{\mathcal S^\complement}\|_{2,1}
\end{align}
Substituting back in (\ref{eq:thm_robust_eq_2}), we have that
\begin{align}
\frac{\|\bold{\tilde v}_{\mathcal S} + \bar{\bold v}_{\mathcal S}\|_{2,1} - \tau\|\bold{\tilde v}_{\mathcal S^\complement}+\bar{\bold v}_{\mathcal S^\complement}\|_{2,1}}
{\| \bold R \bold{\tilde v}\|_2}  \qquad  \qquad &\nonumber
\\
\leq \frac{1+\bar \tau + \tau}{{\tilde\lambda(\bold R^\top \bold R)}^{\sfrac{1}{2}}{|\mathcal V|}^{-\sfrac{1}{2}}
} = \gamma
\label{eq:thm_proof_gamma_bound}
\end{align}
which is the desired inequality.

\item \textit{Case 2:} If $\|\bold{\tilde v}_{\mathcal S^\complement}\|_{2,1} <\|\bar{\bold v}_{\mathcal S^\complement}\|_{2,1}$, then using the fact that $\bar \tau\leq \tau<1$, (\ref{eq:thm_robust_eq_3}) simplifies to
\begin{align}
\|\bold {\tilde v}_{\mathcal S}\|_{2,1} &+ \bar \tau \|\bar {\bold v}_{\mathcal S^\complement}\|_{2,1} - \tau \|\bar {\bold v}_{\mathcal S^\complement}\|_{2,1} + \tau \|\bold {\tilde v}_{\mathcal S^\complement}\|_{2,1}
\nonumber \\
&
\leq
\|\bold {\tilde v}_{\mathcal S}\|_{2,1} 
+ \|\bold {\tilde v}_{\mathcal S^\complement}\|_{2,1} 
+ (\bar \tau - \tau)\|\bar {\bold v}_{\mathcal S^\complement}\|_{2,1}
\\
&
= \| \bold {\tilde v}\|_{2,1}
\end{align}
By substituting in (\ref{eq:thm_robust_eq_2}), it can be seen that inequality (\ref{eq:thm_proof_gamma_bound}) holds in this case as well, which completes the proof.
\end{itemize}
\end{proof}
 
From the expression for $\gamma$ given in (\ref{eq:gamma_def}), it appears as though $\gamma$ should increase with the size of the sensor network, on the order of $|\mathcal V|^{\sfrac{1}{2}}$. However, we note that ${\tilde \lambda (\bold R^\top \bold R)}$ also increases with the size of the network \cite[Thm. 4.4]{jordan2022rigidity}. Determining the exact dependence of $\gamma$ on $|\mathcal V|$ is a possibility for future research.